\documentclass[pra,twocolumn,flushbottom,amsmath,amssymb,showpacs,]{revtex4-1}
\usepackage{graphicx}
\usepackage{dcolumn}
\usepackage{bm}
\usepackage{amsthm}
\usepackage{amsmath}
\usepackage{amssymb}
\usepackage{color}
\newtheorem{algorithm}{Algorithm}
\newtheorem{theorem}{Theorem}

\definecolor{Red}{rgb}{1,0,0}
\def\vec#1{{\bm #1}}

\def\diag{\operatorname{diag}}

\def\rank{\operatorname{rank}}
\def\Span{\operatorname{span}}

\def\A{\mathcal{A}}

\def\C{\mathcal{C}}

\def\E{\mathcal{E}}

\def\H{\mathcal{H}}

\def\M{\mathcal{M}}
\def\N{\mathcal{N}}

\def\ONE{\mathbb{I}}

\def\CC{\mathbb{C}}

\begin{document}

\title{Numerical method for finding decoherence-free subspaces and its applications}
\author{Xiaoting Wang$^{1,2,4}$, Mark Byrd$^3$, and Kurt Jacobs$^{1,4,5}$}

\affiliation{ 
$^1$Department of Physics, University of Massachusetts at Boston, Boston, MA 02125, USA \\
$^2$Research Laboratory of Electronics, Massachusetts Institute of Technology, Cambridge, Massachusetts 02139, USA\\
\mbox{$^3$Physics Department \& Computer Science Department, Southern Illinois University, Carbondale, Illinois 62901, USA}\\ 
$^4$ Advanced Science Institute, RIKEN, Wako-shi 351-0198, Japan \\ 
$^5$Hearne Institute for Theoretical Physics, Louisiana State University, Baton Rouge, LA 70803, USA
} 


\begin{abstract}
In this work, inspired by the study of semidefinite programming for block-diagonalizing matrix *-algebras, we propose an algorithm that can find the algebraic structure of decoherence-free subspaces (DFS's) for a given noisy quantum channel. We prove that this algorithm will work for all cases with probability one, and it is more efficient than the algorithm proposed by Holbrook, Kribs, and Laflamme [Quant. Inf. Proc. {\bf 80}, 381 (2003)]. In fact, our results reveal that this previous algorithm only works for special cases. As an application, we discuss how this method can be applied to increase the efficiency of an  optimization procedure for finding an approximate DFS.  
\end{abstract}


\maketitle
\vspace{-4mm}
\section{Introduction}
\vspace{-3mm}
Decoherence and other noises cause errors in quantum information processing.  Several methods have been proposed to significantly reduce these errors~\cite{PhysRevA.57.3276}, such as quantum error correcting codes~\cite{Shor:95,Calderbank:96,Steane,Gottesman:97,Gottesman:97b,Knill:97b,Gaitan:book}, decoherence-free/noiseless subsystems (DFS's)~\cite{PhysRevLett.79.3306,Lidar98, Lidar99, Zanardi99, Bacon00, Viola00, PhysRevA.63.012301,Wu02,PhysRevLett.91.187903,PhysRevA.72.042303, Bishop-Byrd}, and dynamical decoupling (DD)~\cite{Viola99, Viola99b, Kofman01, Wu02b, Uhrig07, Uhrig09, Souza11}.  However, dynamical decoupling controls often do not satisfy the strict experimental requirements and quantum error correcting codes, while possibly protecting against any error, require a large resource overhead.  Decoherence-free subspaces, or noiseless subsystems, can reduce overhead since they do not require detecting and correcting errors but are difficult to identify and use.  

In terms of the (Krauss) operator-sum representation, the noisy quantum evolution can be fully characterized by a set of operators that generate a noise algebra $\A$. The structure of a DFS can be recovered by studying the algebraic structure of $\A$, or its commutant algebra $\A'$, both of which are special examples of a so-called $\CC^*$-algebra, or matrix $*$-algebra~\cite{Wedderburn}.  Although DFS's can be obtained analytically for certain noisy systems, this is not possible in general. In~\cite{Holbrook} a numerical algorithm is proposed in which the noise algebra $\A$ is used to calculate the commutant algebra $\A'$ and decompose it into the algebraic form which gives the structure of all DFS's if any exist. However, we find that this algorithm is incomplete although useful for the special solutions of the basis of $\A'$ as chosen in~\cite{Holbrook}. In practice, in most of the cases, computers will pick up other solutions in which their algorithm does not give the complete decomposition. Because of this we were motivated to find a new algorithm that provides a general algorithm which can be explicitly implemented to give the required decomposition. 

In fact, since $\A$ and $\A'$ are special examples of a matrix $*$-algebra, we can try to solve the more general problem of how to decompose an arbitrary matrix $*$-algebra. We find that this problem is equivalent to that of simultaneously block-diagonalizing a matrix $*$-algebra, and this has been well-studied in research on semidefinite programming~\cite{Murota, Klerk}. In particular, in~\cite{Murota}, a numerical method was proposed to find the finest block-diagonalization of the algebra generated by real symmetric matrices. Their method consists of two steps. First, decompose the algebra into simple components, and second, decompose each simple component into irreducible components (the details of this procedure will be clarified in the following). In this work, we will show that such a two-step algorithm can also be applied to the algebra generated by Hermitian matrices, such as $\A$ and $\A'$. We will also give analysis and proofs to show the validity of this algorithm. As applications, we apply our new algorithm to the collective-noise model in~\cite{Holbrook}, and compare the numerical results with the algorithm proposed there. We find that our improved algorithm is not only \emph{valid in general}, but is also more \emph{efficient}: it requires fewer conditional loops, and requires only $\A$ or $\A'$ alone, rather than both~\cite{Holbrook}.  

The paper is organized as follows: in Section~\ref{sec:model}, we introduce the matrix *-algebra $\A$ generated by noise operators for a given quantum channel, and the Wedderburn decomposition for a DFS. In Section~\ref{sec:algorithm}, we present an algorithm to transform $\A$ into the Wedderburn form using two steps. Finally in Section~\ref{sec:application}, for the collective-noise model, we numerically implement and compare our algorithm with the one proposed in~\cite{Holbrook}. We will also briefly discuss how our algorithm can be used to find a good initial point for the optimization process in searching for an approximate DFS.

\vspace{-3mm}
\section{Algebraic Structure of a Noise Algebra}\label{sec:model}

\vspace{-3mm}
\subsection{Noise Algebra for a Noisy Quantum Channel} 

\vspace{-3mm}
Let $\rho$ be the density operator of an $n$-dimensional quantum system with Hilbert space $\H$. In real physical systems, the evolution of $\rho$ suffers from noise due to its interaction with the environment. Such noisy evolution can be represented as a quantum channel $\E$: $\rho\to \E(\rho)$, where $\E$ is characterized by a set of operators $\{A_k\}$, $j=1,\ldots,p$:
\begin{align*}
\E(\rho)=\sum_{k=1}^p A_k \rho A_k^\dag=p_0\rho+\sum_{k=2}^p A_k \rho A_k^\dag,\quad \sum_k A_k^\dag A_k=\ONE . 
\end{align*}
This is often referred to as the Kraus operator-sum representation~\cite{Nielsen}, where $\{A_k\}$ can include a Hamiltonian as well as irreversible coupling to a Markovian bath. In the following, we will assume the Hamiltonian $H=0$ and only focus on the noise effect on $\rho$. In this case the $\{A_k\}$ contain information purely about the noise, and are known as noise operators. For many channels, we have $A_1=\sqrt{p_0}\ONE$, where $p_0$ represents the probability that no error occurs. The noisy channel is referred to as \emph{unital} if $\E(\ONE)=\sum_k A_k A_k^\dag=\ONE$. Define the \emph{noise algebra} $\A$ to be the $\CC^*$-algebra, or the matrix $*$-algebra generated by $\{A_k\}$. The definition of a matrix $*$-algebra simply implies that $\A$ is closed under matrix summation, multiplication and $\dagger$-operation. The reason why we introduce the concept of a matrix $*$-algebra is that it can be decomposed into a nice algebraic structure, with details in the following.

\subsection{Wedderburn Decomposition for a DFS}
For a general $\rho$, $\E(\rho)\neq \rho$, and thus the quantum information stored in $\rho$ will not be preserved by the noise. It may be possible, however, find a subspace or subsystem in some space $\H_1\in \H$ such that for $\rho\in \H_1$, $\E(\rho)= \rho$. If so $\H_1$ is called a decoherence free subspace or subsystem (DFS). Notice that for a unital channel $\E$, if $[\rho,A_k]=0$, for all $k$, then $\E(\rho)= \rho$. Hence to locate a DFS it is enough to study the commutant of $\A$, which is defined to be 
\begin{align*}
\A'=\{B|[B,A]=0, \ A\in \A\} , 
\end{align*}
and is also a matrix $*$-algebra. Applying the Wedderburn-Artin theorem\cite{Wedderburn, Gijswijt} to a special case, it can be shown that every matrix $*$-algebra with an identity has the following fundamental structure decomposition~\cite{Barker,Gijswijt}: 
\begin{theorem} \label{thm:decom}
(Wedderburn decomposition) 
Let $\A \subseteq \CC^{n\times n}$ be  a matrix $*$-algebra with an identity. Then there exists a unitary transformation $U$ such that $U^\dag \A U$ has a block-diagonal structure: 
\begin{align*}
U^\dag \A U=\diag(\N_1,\N_2,\cdots,\N_\ell)
\end{align*}
where each $\N_i$ corresponds to a simple subalgebra component. Moreover, $\N_i$ has the following block-diagonal structure: 
\begin{align}\label{eqn:N_i}
\N_i=\{\diag(M_i,\cdots,M_i), \, M_i\in \M_{n_i}\}= \M_{n_i}\otimes \ONE_{m_i}
\end{align}
where $\M_{n_i}$ denotes the $n_i\times n_i$ matrix $*$-algebra over the complex field $\CC$. 
\end{theorem}
Here $\N_i =\M_{n_i}\otimes \ONE_{m_i}$ is an algebra different from $\M_{n_i}\oplus \cdots \oplus \M_{n_i}$. 
Applying Theorem~\ref{thm:decom} to the conjugates $\A$ and $\A'$, we can find some unitary $U$ such that:
\begin{subequations}\label{eqn:alge_decom}
\begin{align}
U^\dag\A U&=\bigoplus_i^\ell \N_i=\bigoplus_i^\ell \M_{n_i}\otimes \ONE_{m_i} , \\
U^\dag\A' U&=\bigoplus_i^\ell \N_i'=\bigoplus_i^\ell \ONE_{n_i}\otimes \M_{m_i} . 
\end{align}
\end{subequations}
Mathematically, each $\N_i$, $i=1,\cdots, \ell$, corresponds to a simple component of $\A$, while the subblock $M_i$ at each diagonal position corresponds to an irreducible component. 

Assume that there exists some $m_i>1$, and call this $m_k$. We can encode an arbitrary $m_k$-dimensional state $\bar\rho$ into $\rho= \ONE_{n_k} \otimes {\bar\rho}\oplus  {0}_{res} \in \A'$ such that $\E(\rho)=\rho$, where ${0}_{res} $ represents the zero density operator on the rest of the Hilbert space with respect to $\ONE_{n_i} \otimes {\bar\rho}$. Hence, if we find the Wedderburn decomposition for $\A$ or $\A'$, then each $\N_i$ with $m_i>1$ corresponds to a decoherence-free subsystem (which reduces to a decoherence-free subspace if $n_i=1$). Moreover, since $\A$ and $\A'$ obey the commutant relation given in Eq.(\ref{eqn:alge_decom}), we do not need both the decompositions for $\A$ and $\A'$; one will suffice. 

\section{Numerical Algorithm to Obtain the Wedderburn Decomposition}\label{sec:algorithm}

To find the Wedderburn decomposition for a quantum channel given by a group of noise operators $\{A_k\}$, it is sufficient to find the unitary transform $U$ such that $\A$ and $\A'$ are simultaneously block-diagonalized into the decomposition in Eq.~(\ref{eqn:alge_decom}). An algorithm to do this for real symmetric $A_k$ is given in~\cite{Murota}. Here we construct an equivalent algorithm that we prove works for Hermitian $A_k$. This is sufficient for our purposes, because while the noise operators $A_k$ need not be Hermitian, we can always replace a non-Hermitian operator $A_j$ with the two Hermitian operators $A_j^{(1)}=A_j+A_j^\dag$ and $A_j^{(2)}=i(A_j-A_j^\dag)$ and still have a generating set for the algebra $\A$. For simplicity we simply assume that all the $A_k$ are Hermitian and form a basis for $\A$. The advantage of choosing an Hermitian basis will be clear in the following analysis. 

The algorithm breaks into two steps: 
\begin{algorithm}
(Wedderburn decomposition) Let $\A \subseteq \CC^{n\times n}$ be  a matrix $*$-algebra with an identity, and $A$ a ``generic'' element of $\A$ (``generic'' is defined below). \\
\textbf{Step 1:} 
Find the unitary transform $V$ such that 
\begin{align}\label{eqn:decom1}
V^\dag A V=\diag(\C_1,\C_2,\cdots,\C_\ell) , 
\end{align}
where each $\C_i$ corresponds to some representation of the simple component $\N_i$ in~(\ref{eqn:alge_decom}).\\
\textbf{Step 2}: Find the local unitary transform $P$ such that within each $\C_i$, $P^\dag V^\dag A VP$ is equal to $\N_i=\M_{n_i}\otimes \ONE_{m_i}$. Then $U\equiv VP$ is the required unitary transform for the Wedderburn decomposition. 
\end{algorithm}

To implement the two steps above one picks a single operator $A \in \A$, and diagonalizes $A$ to find the required decompositions. Due to the decomposition in Eq.~(\ref{eqn:alge_decom}), we know that there exists a unitary transformation $\bar V$ such that for $A \in \A$, 
\begin{align}\label{eqn:generic}
\bar V^\dag A \bar V=\bigoplus_i^\ell (\ONE_{m_i}\otimes D_i) , 
\end{align}
where the $D_i$ are diagonal matrices whose elements are the eigenvalues of $A$. To obtain the spaces spanned by the simple algebras from this eigenvalue decomposition, we need to pick an $A$ such that the $D_1,\ldots, D_{\ell}$ do not share any eigenvalues, and the eigenvalues in each $D_i$ are distinct. Note that $A$ will have this property if it has the maximum possible number of distinct eigenvalues. It can be shown~\cite{Murota} that the set of operators that have this maximum number of distinct eigenvalues is topologically dense in $\A$, and so we will refer to such $A$ as being \textit{generic}. If we randomly choose $A$ from $\A$ using a suitable measure, it will be generic with probability $1$. A simple way to generate a generic $A$ is to choose a random vector $\vec \alpha=(\alpha_1,\ldots, \alpha_k)$ and generate $A=\sum_{j=1}^k \alpha_j A_j$. 

After picking a generic $A$, we diagonalize $A$ and obtain the distinct eigenvalues, $\lambda_j$, and their multiplicities, $k_j$, $j=1,\ldots, q$. We then group these eigenvalues and write down the eigenspace decomposition according to their multiplicities in a non-decreasing order:
\begin{align}\label{eqn:eigen_decom} 
V^\dag A V=\diag(\lambda_1\ONE_{k_1},\lambda_2\ONE_{k_2},\cdots, \lambda_q\ONE_{k_q})
\end{align}
with $V=(V_1,V_2,\ldots,V_q)$ where each $V_j$ is composed of the eigenvectors corresponding to the eigenspace of $\lambda_j$. We can further define a new division of $V$: $V=(K^{(1)},K^{(2)},\cdots, K^{(s)})$ where each $K^{(r)}$ is the union of all eigenspaces $V_j$ with the same multiplicity $p_r$:
\begin{align*}
K^{(r)}=\bigoplus_{k_j=p_r} V_{j}
\end{align*}
where $p_1<p_2<\cdots < p_s$ are the distinct multiplicities of the eigenvalue $\alpha_j$'s.  

By Theorem~\ref{thm:decom}, each $V_j$ will lie in some simple component $\N_i$, $i=1,\ldots, \ell$. Due to the form of $\N_i= \M_{n_i}\otimes \ONE_{m_i}$, we immediately know that only $V_j$'s within the same $K^{(r)}$ can belong to the same $\N_i$, and $V_j$'s in different $K^{(r)}$'s must belong to different $\N_i$'s. Hence, each $K^{(r)}$ must either be some $\N_i$, or a direct sum of a few $\N_i$'s, in which case $K^{(r)}$ can be further decomposed. In either case $\A$ is block-diagonalized over the division $V=(K^{(1)},K^{(2)},\cdots, K^{(s)})$. 

There is a simple method to check whether $K^{(r)}$ can be further decomposed: we choose another randomly generated $\bar A=\sum_{j=1}^k \beta_j A_j \in \A$. Since $\vec \alpha$ and $\vec \beta$ are independent, with probability 1, $A$ and $\bar A$ will generate the whole algebra $\A$. As we have pointed out, both $A$ and $\bar A$ are block-diagonalized over the division $(K^{(1)},K^{(2)},\cdots, K^{(s)})$. If there exist $V_j$ and $V_{j'}$ within some $K^{(r)}$ such that $V_j^\dag \bar A V_{j'}=0$, then $\bar A$ can be further block-diagonalized on $K^{(r)}$ over the division between $V_j$ and $V_{j'}$. In this way, by checking the value of $V_j^\dag \bar A V_{j'}$ between all different $j$ and $j'$ on each $K^{(r)}$, we can identify the structure of each $\N_i$ in each $K^{(r)}$, and finally make both $A$ and $\bar A$ simultaneously block-diagonalized over $\oplus_i \N_i$. Since $A$ and $\bar A$ will generate $\A$ with probability 1, we can claim that the whole algebra $\A$ has been simultaneously block-diagonalized over $\oplus_i \N_i$. However, we should notice that within each sub-block $\N_i$, $\A$ may not be the same as $\M_{n_i}\otimes \ONE_{m_i}$, but some representation of it, so we will instead denote the sub-block by $\C_i$. Thus we have obtained a $V$ that transforms $\A$ into the form of Eq.~(\ref{eqn:decom1}). In particular this $V$ has already transformed the generic $A$ into the Wedderburn form: 
\begin{align}\label{eqn:A_decom}
V^\dag A V=\bigoplus_i^\ell(D_i \otimes \ONE_{m_i}) , 
\end{align}
where $D_i$ is a diagonal matrix with all distinct eigenvalues of $A$ on each $\C_i$. 

Now we note that $V^\dag \bar A V$ is usually not in the form of $\N_i$ on $\C_i$. In the next step, we are looking for a further unitary transform that leaves $V^\dag A V$ invariant but transforms $V^\dag \bar A V$ into the form of $\N_i$ on each $\C_i$. Without loss of generality, let us focus on a simple component $\C_i$ which is composed of a few eigenspaces $V_j$ of $A$: 
\begin{align*}
\C_i= V_{1}^{(i)} \oplus V_2^{(i)}\oplus\cdots \oplus V_{m_i}^{(i)}
\end{align*}
If according to the division $\oplus_i  \C_i$, we define a local unitary transform $P$ to be: 
\begin{align}\label{eqn:pi}
P &\equiv \diag(P^{(1)}, P^{(1)},\cdots,  P^{(\ell)})\\
P^{(i)} &\equiv \diag(P_{1}^{(i)}, P_2^{(i)},\cdots,  P_{m_i}^{(i)})
\end{align}
where $P_j^{(i)}$ is a unitary matrix on the subspace $V_j^{(i)}$, then such $P$ will leave $V^\dag \bar A V$ invariant. Moreover, the following result is proved in Proposition 3.7 in~\cite{Murota}: 

\begin{theorem}\label{thm:local}
For $A$ and $V$ satisfying (\ref{eqn:A_decom}), there exists a local unitary transform $P$ as in (\ref{eqn:pi}) such that $P^\dag V^\dag \A VP=\oplus_i (\M_{n_i}\otimes \ONE_{m_i})$.
\end{theorem}

Hence, it is possible to construct a local unitary $Q$ (which may not be equal to $P$) in the form of Eq.~(\ref{eqn:pi}) such that $\bar Q^\dag V^\dag \A V\bar Q$ is in the Wedderburn form. Before we design the required $Q$, we would like to find out what the matrix of $V^\dag \bar A V$ looks like on $\C_i$ after the transform $V$. Notice that since Theorem~\ref{thm:local} only claims the existence of such $P$, the local transform $Q$ we finally construct may look either the same as, or different from $P$. 

A matrix is called a \emph{scalar matrix} if it is equal to a scalar times an identity matrix. On each $\C_i$, as a corollary of Theorem~\ref{thm:local}, we have 
\begin{align}\label{eqn:VAV}
\bar A_{j,j'}\equiv (V^\dag \bar A V)^{(i)}_{j,j'}=V_j^{(i)\dag} \bar A V^{(i)}_{j'}=k_{j,j'}P_j^{(i)} P_{j'}^{(i)\dag}
\end{align}
For $j=j'$, $\bar A_{j,j'}$ is equal to $k_{j,j}\ONE_{m_i}$ (that is, the diagonal sub-blocks of $(V^\dag \bar A V)^{(i)}$ are already in scalar-matrix form). For $j\neq j'$ the off-diagonal sub-blocks $\bar A_{j,j'}$ may or may not be in this form. Our next goal is to find a local transform $Q$ in the form of (\ref{eqn:pi}) such that $(Q^{(i)\dag} (V^\dag \bar A V)^{(i)}Q^{(i)})_{j,j'}$ are in scalar-matrix form for all $j$ and $j'$.

On each $\C_i$, we can sequentially construct each $Q_j^{(i)}$ in $Q^{(i)} \equiv \diag(Q_{1}^{(i)}, Q_2^{(i)},\cdots,  Q_{m_i}^{(i)})$. First, choose $Q_{1}^{(i)}=\ONE_{m_i}$. Then for $j\ge 2$ define
\begin{subequations}
\label{eqn:q_j}
\begin{align}
\hat Q_{j}^{(i)}&=\big (V_1^{(i)\dag} \bar A V^{(i)}_{j}\big)^{-1} Q_1^{(i)} , \\
Q_{j}^{(i)}&=\frac{1}{||q_j||}\hat Q_{j}^{(i)} , 
\end{align}
\end{subequations}
where $q_j$ is the first row of $\hat Q_{j}^{(i)}$. We now prove that this $Q^{(i)}$ is the required unitary transform. 
 
\begin{theorem}\label{thm:unitary}
$Q_{j}^{(i)}$ as defined in (\ref{eqn:q_j}) are unitary matrices, and $Q^{(i)} \equiv \diag(Q_{1}^{(i)}, Q_2^{(i)},\cdots,  Q_{m_i}^{(i)})$ is the unitary transform such that $Q_j^{(i)^\dag} (V^\dag \bar A V)^{(i)}_{j,j'} Q_{j'}^{(i)}$ are scalar matrices. 
\end{theorem}
\begin{proof}
To show $Q_{j}^{(i)}$ is a unitary matrix, it is sufficient to show $\hat Q_{j}^{(i)}\hat Q_{j}^{(i)\dag}$ is in scalar-matrix form. For $j\ge 2$, from (\ref{eqn:VAV}), we have:
\begin{align*}
\hat Q_{j}^{(i)}\hat Q_{j}^{(i)\dag}&=(k_{1,j}P_1^{(i)} P_{j}^{(i)\dag})^{-1} (k^*_{1,j}P_j^{(i)} P_{1}^{(i)\dag})^{-1}\\
&= |k_{1,j}|^{-2}\big ( P_j^{(i)} P_{1}^{(i)\dag} P_1^{(i)} P_{j}^{(i)\dag} \big )^{-1}= |k_{1,j}|^{-2} \ONE_{m_i}
\end{align*}
Hence, after normalization, $Q_{j}^{(i)}$ becomes a unitary matrix.  In addition, 
\begin{align*}
&Q_j^{(i)^\dag} (V^\dag \bar A V)^{(i)}_{j,j'} Q_{j'}^{(i)}=Q_j^{(i)^\dag}  k_{j,j'}P_j^{(i)} P_{j'}^{(i)\dag} Q_{j'}^{(i)}\\
=&k_{j,j'}/(k_{1,j}^*k_{1,j'})P_1^{(i)} P_{j}^{(i)\dag}P_j^{(i)} P_{j'}^{(i)\dag} P_{j'}^{(i)} P_1^{(i)\dag}\\
=&k_{j,j'}/(k_{1,j}^*k_{1,j'})\ONE_{m_i},
\end{align*}
and so all sub-blocks are in the scalar-matrix form.
\end{proof}

Numerically, following Eq.~(\ref{eqn:q_j}), we can construct the local unitary transform $Q$ in the form of Eq.~(\ref{eqn:pi}) that leave $A$ invariant but transforms $\bar A$ into the form $\oplus_i(M_{n_i}\otimes \ONE_{m_i})$. Since $\A$ is generated by $A$ and $\bar A$, we can claim that the whole algebra $\A$ is in the Wedderburn form after the unitary transform $U\equiv VQ$. We summarize the full algorithm in Table~\ref{tab:1}.
\begin{table}
\begin{ruledtabular}
\begin{tabular}{ll}
\textbf{Step 1}: & (a) from $\A$, pick two generic matrices $A$ and $\bar A$\\ 
& (b) diagonalize $A$ and $\bar A$ to get $V$ as in Eq.~(\ref{eqn:eigen_decom})\\
& (c) find the structure of $\N_i$ in $K^{(r)}$, getting Eq.~(\ref{eqn:A_decom})\\
\hline
\textbf{Step 2}: &(d) build the local transform $Q$ using Eq.~(\ref{eqn:q_j})\\
& (e) $U=VQ$ is the required unitary in Eq.~(\ref{eqn:alge_decom})
\end{tabular}
\end{ruledtabular}
\caption{Algorithm to find $U$ in the decomposition Eq.~(\ref{eqn:alge_decom}). \label{tab:1}}
\end{table}

When implementing the algorithm for a given noisy channel, we can calculate the Wedderburn form for either $\A$ or $\A'$, depending on which one is easier to derive. Notice that for special cases when $\N_i=\M_k\otimes \ONE_1$, or $N_i=\ONE_k$, after Step 1 in Algorithm 1, $\C_i$ will already be the same as $\N_i$. For such cases, there is no need to implement Step 2, and we can simply choose $Q^{(i)}=\ONE$ on $\N_i$.

\section{Applications}\label{sec:application}

\subsection{Finding the DFS Structure of a Channel}

The primary application for the above algorithm is deriving the DFS structure for a given noisy quantum channel, and finding the corresponding unitary transform $U$ in Eq.~(\ref{eqn:alge_decom}). First of all, let's reinvestigate the collective noise model calculated in~\cite{Holbrook}. For a system with $n_q$ qubits, we say a quantum channel $\E$ is under collective noise if 
\begin{align*}
\E(\rho)&=A_x \rho A_x^\dag + A_y \rho A_y^\dag + A_z \rho A_z^\dag , \\
A_k &= \frac{1}{\sqrt{3}} e^{iS_k}, \quad k=x,y,z , 
\end{align*}
where 
\begin{align*}
S_x = \sum_{i=1}^{n_q} X_i, \quad S_y = \sum_{i=1}^{n_q} Y_i,\quad S_z = \sum_{i=1}^{n_q} Z_i 
\end{align*}
are sums of local Pauli operators on each qubit. For such a noisy channel, we can define the algebra generated by the noise operators by 
\begin{align*}
\A &\equiv \Span\{ A_x,A_y,A_z \}= \Span\{ S_x,S_y,S_z \} , \\
&= \Span\{ S_x,S_y \}= \Span\{ S_y,S_z \}= \Span\{ S_x,S_z \} , 
\end{align*}

We would like to find the Wedderburn decomposition of $\A$ or $\A'$ in the form of Eq.~(\ref{eqn:alge_decom}). Notice that since the collective noise channel is a special type of noisy channel, we can actually derive the the fundamental decomposition by using Young diagrams for the
addition of angular momentum for any value of $n_q$ \cite{Byrd:06}. For example, for $n_q=3$, $\A=(\M_2\otimes \ONE_2)\oplus \M_4$; for $n_q=4$, $\A=\ONE_2\oplus (\M_3\otimes \ONE_3)\oplus \M_5$. However, theory does not give a specific basis for the operators $\A$, and must identify one numerically. In the following, we shall apply both the algorithm suggested in~\cite{Holbrook} and the above Algorithm~1 to the collective noise channel and compare the numerical results. 

In the algorithm suggested by~\cite{Holbrook}, we need to first calculate each $B_j$ in $\A'=\Span\{ B_1,\ldots,B_r \}$, where $\{B_j\}$ form a basis for $\A'$. Next, based on the operators $\{B_j\}$, $j=1,\cdots,r$, we find a group of so-called minimal-reducing projectors $P_j$, and then block-diagonalize $\A'$ into the form $\diag(\C_1,\C_2,\cdots,\C_\ell)$, where  $\C_i$ is a representation of $\N_i\equiv \M_{n_i}\otimes \ONE_{m_i}$. Then after reshuffling the order of basis vectors, all the diagonal sub-blocks are transformed into the scalar-matrix form. It is then claimed in~\cite{Holbrook} that the whole algebra $\A'$ is in the form of Eq.~(\ref{eqn:alge_decom}). 

If we compare the algorithm in~\cite{Holbrook} with the algorithm in Table~\ref{tab:1}, we find that the former achieves Step 1, but Step 2 is missing. Step 2 turns out to be necessary for most cases, since the solution $\{B_j\}$ as the set of basis $\A'$ is not unique. It is true that for the $\{B_j\}$ chosen in~\cite{Holbrook}, Step 1 and reshuffling of basis vectors are enough to transform $\A'$ into the form of (\ref{eqn:alge_decom}), but such choice of $\{B_j\}$ is very special. A different solution for the $\{B_j\}$ will be obtained it is derived numerically by solving the system of linear equations $[B_j,A_k]=0$, $k=1,\cdots, p$. 

As an example of the necessity of step 2, consider for $n_q=4$, and $\A'=\M_2\oplus (\M_3\otimes \ONE_3)\oplus \ONE_5$. Following the algorithm in~\cite{Holbrook} and using Matlab, we derive a group of 14 orthonormal basis matrices in $\A'$ ($\{B_j\}$, $j=1,\cdots,14$) which are different from those in~\cite{Holbrook}. Then we find the corresponding minimal-reducing projectors $P_m$, $m=1,\cdots,6$, in which $\rank(P_1)=\rank(P_2)=1$ corresponding to the $\N_1=\M_2$ subspace, $\rank(P_6)=5$ corresponding to the $\N_3=\ONE_5$ subspace,  and $\rank(P_i)=3$, $i=3,4,5$, corresponding to the $\N_2=\M_3\otimes \ONE_3$ subspace. Hence, $\{P_m\}$ induces a unitary transform $V$ such that  $V^\dag \A' V$ is in the form $\diag(\C_1,\C_2,\C_3)$, where $\C_1=\N_1$, $\C_3=\N_3$, and $\C_2$ is some representation of $\N_2$. After reshuffling the basis vectors we can make the three diagonal sub-blocks of $V^\dag B_j V$ on $\C_2$ in the scalar-matrix form. 

Specifically, if we still denote $V^\dag B_1V$ as $B_1$, then after reshuffling, on $\C_2$ we have
\begin{align*}
B_1=\begin{pmatrix}
-0.167 \ONE_3  & B_{1,2} & B_{1,3}\\        
B_{1,2}^\dag &  0.233\ONE_3&  B_{2,3} \\
B_{1,3}^\dag & B_{2,3}^\dag & 0.308\ONE_3 , 
\end{pmatrix}
\end{align*}
where
\begin{widetext}
\begin{align*}
B_{1,2}=\begin{pmatrix}
-0.180 + 0.089i & -0.120 - 0.242i  &-0.042 + 0.063i\\
0.097 + 0.169i &  0.103 - 0.001i & -0.259 + 0.059i\\
0.193 - 0.060i&  -0.100 - 0.159i & -0.051 - 0.200i
\end{pmatrix} , \\ 
B_{1,3}=\begin{pmatrix}
-0.074 + 0.136i &  0.055 + 0.047i &  0.072 + 0.055i\\
  -0.003 + 0.087i&  -0.094 + 0.080i & -0.074 - 0.095i\\
  -0.066 + 0.039i &  0.087 - 0.0960i&  -0.114 - 0.041i
\end{pmatrix} , \\
B_{2,3}=\begin{pmatrix}
-0.030 - 0.078i&  -0.020 - 0.068i&  -0.067 + 0.180i\\
  -0.118 + 0.099i &  0.139 + 0.050i&  -0.037 + 0.045i\\
  -0.133 - 0.023i&  -0.016 - 0.147i & -0.009 - 0.093i
\end{pmatrix} . 
\end{align*}
Therefore, we see that for a general solution $\{B_j\}$ for $\A'$, such as the solution derived from the matlab routine, the algorithm suggested by~\cite{Holbrook} fails to give the Wedderburn decomposition, although it does give the the correct form for the particularly chosen $\{B_j\}$ in~\cite{Holbrook}.  Hence, in practice, the algorithm in~\cite{Holbrook} is sometimes insufficient, which motivated us to develop the modified algorithm presented here. Next, we would like to continue with this example for $n_q=4$, and following our algorithm to find the local unitary transform $Q=\diag(Q_1,Q_2,Q_3)$ such that all the sub-blocks of $Q^\dag B_j Q$ are in the scalar matrix form. Strictly speaking, we should use the random combination method in the last section to pick up a generic $\bar B$ for Step 2. However, for this particular example, the above $B_1$ is already a generic matrix, so we will instead based on $B_1$ to construct $Q$.

Define $Q_1=\ONE_3$, and according to Step 2 in Table~\ref{tab:1}, we define $\hat Q_k=B_{1,k}^{-1}$ and $Q_k=1/||q_k||\hat Q_k$, $k=2,3$, where $q_k$ is the first row of  $\hat Q_k$. Then we have 
\begin{align*}
Q^\dag B_1 Q=\begin{pmatrix}
-0.167 \ONE_3  & 0.345\ONE_3 &0.1931\ONE_3\\        
0.345\ONE_3 &  0.233\ONE_3&  -0.033 - 0.219i\ONE_3 \\
0.193\ONE_3 & -0.033 + 0.219i\ONE_3 & 0.308\ONE_3
\end{pmatrix}
\end{align*}
\end{widetext}

We can double-check the form of $Q^\dag B_j Q$ for other $B_j$ and we will find that all $B_j$ are transformed in the form of $\M_3\otimes I_3$. Hence, for this particular set of $\{B_j\}$, we have explicitly constructed the unitary matrix $U=VQ$ that transforms $\A'$ into the Wedderburn decomposition, which is not accessible from the algorithm in~\cite{Holbrook}. Now we know how to encode an arbitrary three-level quantum state $\bar \rho$ into the DFS $\N_2$. In the current computational basis, the encoded density operator $\rho$ should take the form:
\begin{align*}
\rho=U\big(0\oplus (\bar\rho\otimes \ONE_3)\oplus 0\big)U^\dag
\end{align*}

We note finally that: i) in the above implementation of our algorithm, we have skipped Step 1 since the algorithm in~\cite{Holbrook} has already transformed all $B_j$ into \ref{eqn:decom1}. ii) We have applied Step 2 to $B_1$ instead of to a random combination of the $\{B_j\}$, so as to make it easy for comparison. In practice, we do not really have to use the two generic operators as we suggest. They are introduced primarily to guarantee the validity of the algorithm. 

\subsection{Searching for an Approximate DFS}

Although for every noisy channel there exists a decomposition as in (\ref{eqn:alge_decom}), not all channels have a DFS that is useful for protecting quantum information. In fact, the noisy channels with a useful DFS constitute only a very small set. For many channels the algebraic decomposition (\ref{eqn:alge_decom}) looks like the following:
 \begin{align*}
\A'=\bigoplus_i k_i \ONE_{m_i} , 
\end{align*}
where $k_i\neq k_j,$ for $i\neq j$, which means that all the $\M_{n_i}$'s are 1-dimensional and so cannot store quantum information. When this happens, we would like to ask an alternative question: does there exist a subsystem on which the noise, even if not zero, is significantly reduced? This is the concept of an {\it approximate} DFS (ADFS).  

It is not easy to characterize an ADFS by algebraic conditions, as we have done for a perfect DFS. Rather, an ADFS should be formulated as an optimal solution such that the noise on the system is reduced as much as possible. Hence, it is possible to obtain an ADFS numerically by solving the corresponding optimization problem. Since there is more than one way to quantify the effects of noise, there is more than one way to define the function to be minimized in the optimization.  For the purposes of our analysis in what follows, we simply assume that one such function has been chosen, and denote it by $J$. Furthermore, one must specify when the noise is ``small enough" to be helpful as an ADFS.  

The problem of finding an ADFS involves searching for the optimal unitary matrix $U$ that transforms the original basis into a new basis, such that a state $\rho_1$, encoded in $\rho=\frac{1}{n_2}\rho_1\otimes \ONE_{n_2}\oplus 0_{res}$, experiences the least noise under the noise operators $\{\bar A_k\}$, where $\bar A_k=UA_kU^\dag$. That is, we want to minimize $J(U)$, where $U$ varies over the unitary group. Numerically, we can apply the BFGS quasi Newton method for the optimization~\cite{Nocedal06}. Broadly speaking this involves i) choosing  an initial point $U=U^{(0)}$ in the unitary matrix space; ii) calculating the value and the gradient of the objective function $J^{(0)}=J[U^{(0)}]$; iii) using the value and the gradient to implicitly derive the Hessian information and use them all to get a new $U^{(1)}$ such that $J[U^{(1)}]<J[U^{(0)}]$. Repeating steps ii) and iii), we obtain a sequence of $\{U^{(k)}\}$ with a limit corresponding to a local minimum of $J$. The limiting unitary transform $\bar U$ is what we are looking for.  
 
Due to the existence of many local minima, different initial choices of $U^{(0)}$ may result in different values of the optimized $J$. In particular, as the dimension gets larger we may need to run the optimization many times before we obtain a value of $J$ close to the global minimum. Hence a wise choice of $U^{(0)}$ can be very important in performing the numerical optimization. One important result in optimization theory is that any gradient-based algorithm only guarantees that the iteration sequence will converge to some local minimum; however, if the initial point of optimization iteration is very close to the global minimum, then the iteration sequence will converge to the global minimum. On the other hand, if our noisy quantum channel $\E$ can be considered as a perturbation of another channel $\E'$ that has a perfect DFS, then the ADFS of $\E$ should be pretty close to the DFS of $\E$. Hence, we can apply the DFS algorithm in Section~\ref{sec:algorithm} to first calculate the unitary matrix $U_0$ for the DFS of $\E'$, and then run the optimization for ADFS, with $U^{(0)}=U_0$. In that way, we will be able to derive the ADFS more efficiently.  
 
Specifically, let us take the collective noise model as an example, but this time add a perturbation to the original noise operators:
 \begin{align*}
\tilde A_x=V_{\epsilon}A_x, \quad \tilde A_y=A_y,\quad \tilde A_z=A_z , 
\end{align*}
where we have defined a perturbation unitary matrix $V_{\epsilon}$ that is sufficiently close to the identity: $||V_{\epsilon}-\ONE||<\epsilon$. Under the new noise operators $\tilde A_{x,y,z}$, we apply the algorithm in Section~\ref{sec:algorithm} and find that there is no useful DFS. Next we try to run optimization to this model in searching for ADFS. First we do the optimization using random initial point. For example, choosing $\epsilon=1$ and $n_q=4$, we run the optimization routine $6$ times, starting from different initial $U^{(0)}$, and record the final minimized $J_{min}$ in Table~\ref{tab:2}.
\begin{table}
\begin{tabular}{|c|c|c|c|c|c|c|}
\hline
 $N$    &     $1$   & $2$   & $3$ & $4$  &$5$ & $6$\\\hline
 $J_{min}$     &  0.2195  &  0.2078 &   0.2088 &   0.0123  &  0.4788 & 0.0125 \\\hline
\end{tabular}
\caption{The final value of $J_{min}$ obtained from numerical optimization, using $6$ different random starting points enumerated by $N$. \label{tab:2}} 
\end{table} 
We see that among the six different runs, only two of them have obtained $J_{min}<0.0125$. Hence we cannot guarantee that we have the best minimized $J$ from a single run of the optimization process from an arbitrary random initial $\{U^{(0)}\}$. However, if we instead choose $U^{(0)}=U_0$, where $U_0$ is the unitary matrix in the Wedderburn decomposition for the perfect DFS of $\{A_{x,y,z}\}$, then the optimization generates the minimized $J_{min}=0.0123$. For other values of $n_q$ we find similar results. Thus our DFS-finding algorithm helps in finding good initial points for the optimization of ADFS searching. 

In addition, many local minima may also result in options for our ADFS implementation.  Whereas our algorithm gives the dimensionally optimal DFS, there may be several ADFSs and the "best" one may not be the dimensionally optimal one.  The ``best" might be the one which has robust, experimentally available controls, or one that has the lowest error rate per unit time.  

\section{Conclusion}

In this work, for a given noisy quantum channel, we have presented an algorithm to numerically calculate the unitary matrix that transforms the original noise algebra into the Wedderburn form, and this gives the structure of all DFS's if any exist. This algorithm is based on the theory of the Wedderburn decomposition of matrix *-algebras. We also compared our algorithm with the earlier algorithm proposed in \cite{Holbrook}, which we found was incomplete. The new algorithm is also more efficient, in that it requires fewer checks and evaluations, and requires only information from either the noise algebra $\A$ or its conjugate $\A'$, rather than both. As an application, we show that the DFS-finding method is helpful in locating good initial points for finding approximate DFS's, and this is likely to be a more practical use for the algorithm.  

\section*{Acknowledgements}

KJ is partially supported by the NSF project PHY-1005571, and KJ and XW are partially supported by the NSF project PHY-0902906 and the ARO MURI grant W911NF-11-1-0268. All the authors are partially supported by the Intelligence Advanced Research Projects Activity (IARPA) via Department of Interior National Business Center contract number D11PC20168. The U.S. Government is authorized to reproduce and distribute reprints for Governmental purposes notwithstanding any copyright annotation thereon. Disclaimer: The views and conclusions contained herein are those of the authors and should not be interpreted as necessarily representing the official policies or endorsements, either expressed or implied, of IARPA, DoI/NBC, or the U.S. Government. 


\bibliography{numer_algorithm_dfs_v4}

%
%
%
%
%
%
%

\end{document}